\documentclass[a4paper,10pt,twoside]{article}

\usepackage{amsmath,amssymb,amsthm}
\usepackage{enumerate}
\usepackage{tikz}
\usepackage{caption}
\usepackage{subcaption}

\makeatletter
\let\@afterindentfalse\@afterindenttrue
\@afterindenttrue
\makeatother

\newtheorem{theorem}{Theorem}
\newtheorem{lemma}{Lemma}
\newtheorem{corollary}{Corollary}
\newtheorem{remark}{Remark}
\newtheorem{conjecture}{Conjecture}
\newtheorem{algorithm}{Algorithm}

\newcommand{\R}{\mathbb{R}}
\newcommand{\C}{\mathbb{C}}
\newcommand{\K}{\mathbb{K}}
\newcommand{\Ran}{\mathrm{Ran}}
\newcommand{\Q}{\mathcal{Q}}

\newcommand{\la}{\lambda}
\newcommand{\dint}{\,\mathrm{d}}
\newcommand{\Ga}{\Gamma}
\newcommand{\E}{\mathop{\mathrm{E}}\nolimits}
\newcommand{\e}{\mathop{\mathrm{e}}\nolimits}
\newcommand{\M}[1]{\mathcal{M}_{#1}}
\newcommand{\MR}[1]{\mathcal{M}^{\mathbb{R}}_{#1}}
\newcommand{\MC}[1]{\mathcal{M}^{\mathbb{C}}_{#1}}

\newcommand{\Tr}{\mathop{\mathrm{Tr}}\nolimits}

\newcommand{\Span}{\mathop{\mathrm{Span}}\nolimits}
\newcommand{\dg}{\mathop{\mathrm{Diag}}\nolimits}
\title{
Volume of the space of qubit channels and some new
results about the distribution of the quantum Dobrushin
coefficient\thanks{
quantum channel, volume, contraction coefficient;
MSC2010: 81P16, 81P45, 94A17
}}

\author{Attila Lovas\thanks{lovas@math.bme.hu}, 
        Attila Andai\\
	Department for Mathematical Analysis,\\
	Budapest University of Technology and Economics,\\
	H-1521 Budapest XI. Stoczek u. 2, Hungary
}

\date{\today}

\begin{document}

\maketitle

\begin{abstract}
% Background
The simplest building blocks for quantum computations are
  the qbit-qbit quantum channels.
In this paper we analyse the structure of these channels
  via their Choi representation.
The restriction of a quantum channel to the space of classical
  states (i.e. probability distributions) is called the 
  underlying classical channel.
The structure of quantum channels over a fixed classical channel is studied,
  the volume of general and unital qubit channels over real and complex state 
  spaces with respect to the Lebesgue measure is computed and explicit 
  formulas are presented for the distribution of the volume of quantum 
  channels over given classical channels.
Moreover an algorithm is presented to generate uniformly distributed channels 
  with respect to the Lebesgue measure, which enables further studies.
With this algorithm the distribution of trace-distance contraction
  coefficient (Dobrushin) is investigated numerically by Monte-Carlo 
  simulations, which leads to some conjectures and points out the
  strange behaviour of the real state space.
\end{abstract}

\section*{Introduction}

% State space
In quantum information theory, a qubit is the quantum analogue 
  of the classical bit. 
A qubit can be represented by a $2\times 2$ 
  self-adjoint positive semidefinite matrix with trace one
\cite{NielsenChuang,PetzQinf,RuskaiSzarekWerner}.
The space of qubits with real entries is
  denoted by $\MR{2}$ and with complex entries
  by $\MC{2}$ respectively. If we do not want to
  emphasise the underlying field, then we just write $\M{2}$.
A linear map $Q:\M{2}\to \M{2}$ is called a qubit channel if
  it is a completely positive and trace preserving (CPT) map \cite{PetzQinf}. 
A qubit channel is said to be unital if it preserves the identity.
Choi has published a tractable representation for completely positive linear 
  maps \cite{Choi}.
To a linear map $Q:\K^{2\times 2}\to \K^{2\times 2}$ 
($\K = \R,\C$) a block matrix
\begin{equation}\label{eq:blk}
\left(
\begin{array}{cc}
 Q_{11} & Q_{12} \\ 
 Q_{21} & Q_{22}
\end{array}
\right)
\quad
 Q_{11},Q_{12},Q_{21},Q_{22}\in\K^{2\times 2}
\end{equation}
  is associated, which is called the \emph{Choi matrix}, 
  such that the action of $Q$ is given by
\begin{equation*}
 \left(
 \begin{array}{cc}
  a & b\\
  c & d
 \end{array}
 \right)\mapsto aQ_{11} + bQ_{12} + cQ_{21} + dQ_{22}.
\end{equation*}
Choi's theorem states that the linear map $Q:\K^{2\times 2}\to\K^{2\times 2}$
  is completely positive if and only if its Choi matrix is positive definite
  \cite{Choi}. 
Hereafter, we will use the same symbol for the qubit channel and its 
  Choi matrix. 
Let 
\begin{equation*}
Q = \left(
\begin{array}{cc}
 Q_{11} & Q_{12} \\ 
 Q_{21} & Q_{22}
\end{array}
\right)
\end{equation*}
  be a qubit channel and we define the \emph{underlying classical channel} as
  the restriction of $Q$ to the space of classical bits 
  (i.e. diagonal matrices). 
The following Markov chain transition matrix can be associated to the 
  underlying channel of the qubit channel $Q$
\begin{equation*}
 P =
 \left(
 \begin{array}{c}
  \dg (Q_{11}) \\
  \dg (Q_{22})
 \end{array}
 \right),
\end{equation*}
  where $\dg (Q_{ii})$ denotes the diagonal of the submatrix $Q_{ii}$ 
  in a row vector.

Like many other quantities of interest in quantum information
  theory the trace distance between states contracts 
  under the action of quantum channels. When $Q$ is a CPT map, 
  we can define the trace-distance contraction coefficient as
\begin{equation*}
 \eta^{\Tr} (Q) = \sup
 \left\lbrace
 \frac{\Tr{|Q(\rho)-Q(\sigma)|}}{\Tr{|\rho-\sigma|}} :
 \rho,\sigma\in\M{2}
 \right\rbrace
\end{equation*}
  which describes the maximal contraction under $Q$. 
This can be regarded as the quantum analogue of the Dobrushin coefficient
  of ergodicity \cite{Dobrushin} and has important applications to
  the problem of mixing time bounds of (quantum) Markov processes, 
  as demonstrated in e.g, \cite{Kemperman,Cohen,Temme,Kastoryano}.
To compute the volume of qubit channels and their distributions
  over classical channels, we use the strategy that was applied by 
  A. Andai to compute the volume of the quantum mechanical state 
  space over $n$-dimensional real, complex and quaternionic Hilbert 
  spaces with respect to the canonical Euclidean measure \cite{AndaiVol}.

% Structure of the paper
The paper is organized as follows.
In the first section we fix the notations for further computations and we   
  mention some elementary lemmas which will be used in the sequel.
In Section 2, the volume of general and unital qubit channels over real and  
  complex state spaces with respect to the canonical Euclidean measure
  are computed and explicit formulas are given for the distribution of the
  volume over classical channels.
Section 3 deals with the distribution of the trace-distance contraction 
  coefficient. 
Cumulative distribution function of $\eta^{\Tr}$ was calculated 
  by Monte-Carlo method on the whole space. 
Supremum of $\eta^{\Tr}$ over a fixed classical channel was calculated 
  explicitly.
As to the infimum of $\eta^{\Tr}$ over a fixed classical channel
  we conjecture that it coincides with the trace-distance contraction
  coefficient of the considered classical channel.
Our conjecture was been confirmed by numerical simulations for unital channels. 
A kind of anomaly observed in the behaviour of $\eta^{\Tr}$ over a fixed  
  classical channel in case of real unital channels.

\section{Basic lemmas and notations}

The following lemmas will be our main tools, we will use them without 
  mentioning, and we also introduce some notations which will be used in the 
  sequel.

The first four lemmas are elementary propositions in linear algebra.
For an $n\times n$ matrix $A$ we set $A_{i}$ to be the left upper $i\times i$ 
  submatrix of $A$, where $i=1,\dots,n$.

\begin{lemma}
The $n\times n$ self-adjoint matrix $A$ is positive definite if and only if 
  the inequality $\det(A_{i})>0$ holds for every $i=1,\dots,n$.
\end{lemma}

\begin{lemma}
The $n\times n$ self-adjoint matrix $A$ is positive definite if and only if 
$U^\ast{AU}$ is positive definite for all unitary matrix $U$.
\end{lemma}

\begin{lemma}
Assume that $A$ is an $n\times n$ self-adjoint, positive definite matrix 
  with entries $(a_{ij})_{i,j=1,\dots, n}$ and the vector $\alpha$ consists of 
  the first $(n-1)$ elements of the last column, that is 
  $\alpha=(a_{1,n},\dots,a_{n-1,n})$.
Then for the matrix $T=\det(A_{n-1}) (A_{n-1})^{-1}$ we have
\begin{equation*}
\det(A)=a_{nn}\det(A_{n-1})-\left< \alpha,T\alpha \right>.
\end{equation*}
\end{lemma}
\begin{proof}
The statement comes from elementary matrix computation, one should expand 
  $\det(A)$ by minors, with respect to the last row.
\end{proof}

\begin{lemma}\label{lem:jac}
Let $A$ be an $n\times n$ invertible matrix and for $1\le k \le n$
  define the complementary minor to $\left(A^{-1}\right)_k$
  as the $(n-k)$-rowed minor obtained from $A^{-1}$ by deleting
  all the rows and columns associated with $A_k$. 
If $(A^{-1})_{k+1,\ldots,n}$ denotes the complementary minor to 
  $\left(A^{-1}\right)_k$, then it is true that
\begin{equation*}\label{eq:jac}
 \det ((A^{-1})_{k+1,\ldots,n}) = \frac{\det (A_k)}{\det (A)}.
\end{equation*}
\end{lemma}
Note that the previous lemma is the special case of 
  Jacobi's theorem \cite{Grandshteyn}. 
We will apply it in the following form.
\begin{corollary}
If $A$ is an $n\times n$ invertible matrix, then
  for the matrix $T=\det (A) (A^{-1})$ we have
\begin{equation*}\label{eq:det}
 \det ((T)_{k+1,\ldots,n}) = \det (A_k) \det (A)^{n-1-k} 
\end{equation*}
  for every $1\le k \le n$.
\end{corollary}

The next two lemmas are about some elementary properties of the gamma function
  $\Gamma$ and the beta integral.
\begin{lemma}
Consider the function $\Ga$, which can be defined for $x\in\mathbb{R}^{+}$ as
\begin{equation*}
\Ga(x)=\int_{0}^{\infty}t^{x-1}\e^{-t}\dint t.
\end{equation*}
This function has the following properties for every natural number $n\neq 0$ 
  and real argument $x\in\mathbb{R}^{+}$.
\begin{align*}
&\Ga(n)=(n-1)!\quad \Ga(1+x)=x\Ga(x)\quad \Ga(1/2)=\sqrt{\pi}\\
&\Ga(n+1/2)=\frac{(2n-1)!!}{2^{n}}\sqrt{\pi}\quad \Ga(n/2)
  =\frac{(n-2)!!}{2^{\frac{n-1}{2}}}\sqrt{\pi}
\end{align*}
\end{lemma}

\begin{lemma}
For parameters $a,b\in\mathbb{R}^{+}$ and $t\in\mathbb{R}^{+}$ the integral 
  equalities
\begin{align*}
         &\int_{0}^{t} x^{a}(t-x)^{b}\dint x
  =t^{1+a+b}\frac{\Ga(a+1)\Ga(b+1)}{\Ga(a+b+2)}\\
G_{a,b}:=&\int_{0}^{1} x^{a}(1-x^{2})^{b}\dint x
  =\frac{1}{2}\frac{\Ga(b+1)\Ga\left(\frac{a+1}{2}\right)}
  {\Ga\left(\frac{a}{2}+b+\frac{3}{2}\right)}
\end{align*}
  hold.
\end{lemma}
\begin{proof}
These are consequences of the formula below for the beta integral
\begin{equation*}
\int_{0}^{1}x^{p}(1-x)^{q}\dint x=\frac{\Ga(p+1)\Ga(q+1)}{\Ga(p+q+2)}.
\end{equation*}
\end{proof}

\begin{lemma}
The surface $F_{n-1}$ of a unit sphere in an $n$ dimensional space is
\begin{equation*}
F_{n-1}=\frac{n\pi^{\frac{n}{2}}}{\Ga\left(\frac{n}{2}+1 \right)}.
\end{equation*}
\end{lemma}
\begin{proof}
It follows from the well-known formula for the volume of the sphere in $n$ 
  dimension with radius $r$
\begin{equation*}
V_{n}(r)=\frac{r^{n}\pi^{\frac{n}{2}}}{\Ga\left(\frac{n}{2}+1 \right)},
\end{equation*}
  since $F_{n-1}=\left.\frac{\dint V_{n}(r)}{\dint r}\right\vert_{r=1}$.
\end{proof}

When we integrate on a subset of the Euclidean space we always integrate with 
  respect to the usual Lebesgue measure.
The Lebesgue measure on $\mathbb{R}^{n}$ will be denoted by $\la_{n}$.
The following lemma is the backbone of our investigations.

\begin{lemma}\label{lem:integral}
Assume that $T$ is an $n\times n$ self-adjoint, positive definite matrix,
  $l\in\R$ and $\mu>0$. 
Let $L$ be an $m$-dimensional subspace of the vector space $\mathbb{K}^{n}$
  and $x$ is a fixed vector. 
Let us denote the orthogonal projection onto the orthogonal complement 
  of the subspace $T(L)$ by $P_{M^\perp}$.
Set
\begin{align*}
&\E^\R (T,\mu,L,x):=\left\{y\in L \vert\ \left<x+y,T(x+y)\right><\mu \right\},
  \quad T_{ij}\in \mathbb{R};\\
&\E^\C (T,\mu,L,x):=\left\{y\in L \vert\ \left<x+y,T(x+y)\right><\mu \right\},
  \quad T_{ij}\in \mathbb{C};
\end{align*} 
  then
\begin{equation*}
 \int\limits_{\E^\R (T,\mu,L,x)}
 (\mu-\left< x+y,T(x+y)\right>)^{l}\dint\la_{m}(y)=
 \frac{F_{m-1}G_{m-1,l}}{\sqrt{\det (T|_L)}}(\mu-||z_0||^2)_+^{\frac{m}{2}+l}
\end{equation*}
  and
\begin{equation*}
 \int\limits_{\E^\C (T,\mu,L,x)}
 (\mu-\left< x+y,T(x+y)\right>)^{l}\dint\la_{2m}(y)=
 \frac{F_{2m-1}G_{2m-1,l}}{\det (T|_L)}(\mu-||z_0||^2)_+^{m+l},
\end{equation*}
  where $T|L$ is the restriction of $T$ to the subspace $L$ and 
  $z_0 := P_{M^\perp}\sqrt{T}x$.
\end{lemma}
\begin{proof}
We prove the statement for the real case only, the other cases can be proved 
  in the same way.
The matrix $T$ is supposed to be positive definite thus there exists
  a unique self-adjoint positive definite matrix $\sqrt{T}$ for which 
  $T = (\sqrt{T})^2$ holds.

Consider the map $\Phi :L\to\R^n$, 
$\Phi (y):=\frac{1}{\sqrt{\mu}}\sqrt{T} (x+y)$ and choose
  an orthonormal basis of the subspace $L$: $e_1,\ldots,e_m$. 
The corresponding parametrization of $\Ran(\Phi)$ is 
\begin{equation*}
z(y_1,\ldots,y_m)
  =\frac{1}{\sqrt{\mu}}\sqrt{T} \left(x+\sum\limits_{i=1}^m y_i e_i \right)
\end{equation*}
  and the induced metric on $\Ran(\Phi)$ can be written as 
\begin{equation*}
g_{ij}
=\left<\frac{\partial z}{\partial y_i},\frac{\partial z}{\partial y_j}\right>
=\left<\frac{1}{\sqrt{\mu}}\sqrt{T} e_i,\frac{1}{\sqrt{\mu}}\sqrt{T} e_j\right> 
=\frac{1}{\mu}\left<e_i,T e_j\right>
\end{equation*}
  hence the inverse Jacobian of this transformation is
  $\frac{\mu^\frac{m}{2}}{\sqrt{\det (T|_L)}}$.
We can write 
\begin{align*}
\int\limits_{\E^\R (T,\mu,L,x)}&
 (\mu-\left< x+y,T(x+y)\right>)^{l}\dint\la_{m}(y)=\\
&=\frac{\mu^{\frac{m}{2}+l}}{\sqrt{\det (T|_L)}}
\int\limits_{\Phi (\E^\R (T,\mu,L,x))}
(1-||z||^2)^l\dint\la_{m}(z).
\end{align*}
The set $\Phi (\E^\R (T,\mu,L,x))$ is the intersection of the affine
  subspace $\Ran (\Phi)$ and the unit ball of $\R^n$ centered at the origin   
  (Figure \ref{fig:1}).
Note that, $\Phi (\E^\R (T,\mu,L,x))$ is non-empty if and only if
  the distance of $\Ran (\Phi)$ from the origin is less that one:
  $d^2:=\frac{1}{\mu}||z_0||^2 = \frac{1}{\mu}||P_{M^\perp}\sqrt{T}x||^2<1$.
\begin{figure}[!ht]
\centering
\begin{tikzpicture}
 \draw[->] ({-(1+sqrt(5))},0) -- ({1+sqrt(5)},0);
 \draw[->] (0,{-(1+sqrt(5))}) -- (0,{1+sqrt(5)});
 \draw     (0,0) circle (2);
 \draw[very thick,->] (0,0)  -- ({sqrt(2)},0) node[pos=0.5,below] {$\frac{1}{\sqrt{\mu}}z_0$};
 \draw[very thick]     ({sqrt(2)},{-sqrt(2)}) -- ({sqrt(2)},{sqrt(2)}) node[pos=0,right] {$\Phi (\E^\R (T,\mu,L,x))$};
 \draw[very thick,->] (0,0) -- ({sqrt(2)},1) node[pos=0.5,sloped,above] {$z$}; 
 \draw (2.2,0) node[below] {$1$};
 \draw ({sqrt(2)},0) -- ({sqrt(2)},1) node[pos=0.5,right] {$r$};
 \draw     ({sqrt(2)},{-(1+sqrt(5))}) -- ({sqrt(2)},{1+sqrt(5)}) node[pos=0,right] {$\Ran (\Phi)$};
\end{tikzpicture}
\caption{The sketch of the region of integration.}
\label{fig:1}
\end{figure}
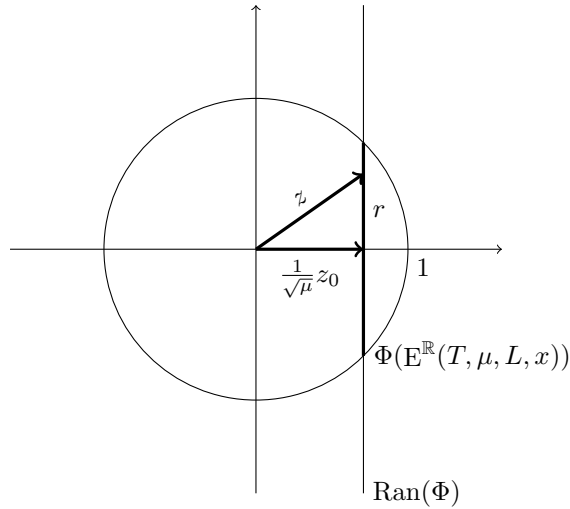
Then we compute the integral in spherical coordinates. 
The integral with respect to the angles gives the surface of the sphere
  $F_{m-1}(1-d^2)_{+}^{(m-1)/2}$ and the radial part is
\begin{equation*}
 \frac{\mu^{\frac{m}{2}+l}}{\sqrt{\det (T|_L)}}
 F_{m-1}(1-d^2)_+^{\frac{m-1}{2}+l}
\int\limits_0^{\sqrt{1-d^2}}
\left(1-\frac{r^2}{1-d^2}\right)^l
\left(\frac{r}{\sqrt{1-d^2}}\right)^{m-1}
\dint r
\end{equation*}
We substitute $u = \frac{r}{\sqrt{1-d^2}}$ and obtain the desired formula.
\end{proof}

\begin{remark}\label{rem:1}
If $A$ is an $n\times n$ positive definite matrix, $L\subseteq\K^n$ is a 
 subspace and $M=\sqrt{A^{-1}}(L)$, then $M^\perp = \sqrt{A}(L^\perp)$
 because
\begin{equation*}
M^\perp = (\sqrt{A^{-1}}(L))^\perp 
  = \ker (P_L \sqrt{A^{-1}})=\sqrt{A}(L^\perp).
\end{equation*}
\end{remark}

Recall that the Pauli matrices 
$\sigma_1=\left(\begin{array}{cc} 0 & 1  \\ 1 & 0  \end{array} \right)$, 
$\sigma_2=\left(\begin{array}{cc} 0 & -i \\ i & 0  \end{array} \right)$ and
$\sigma_3=\left(\begin{array}{cc} 1 & 0  \\ 0 & -1 \end{array} \right)$ 
  together with 
$I=\left(\begin{array}{cc} 1 & 0 \\ 0 & 1 \end{array} \right)$
  form an orthogonal basis of the space of $2\times 2$ self-adjoint matrices.

\section{The volume of qubit channels}

To determine the volumes of different qubit quantum channels we use the same
  method which consist of three parts.
First we use an unitary transformation to represent channels in a suitable 
  form for further computations.
Then we split the parameter space into lower dimensional parts such that
  the adequate application of the previously mentioned lemmas leads us to
  the result.

\subsection{General qubit channels}

A block matrix $Q$ of the form \eqref{eq:blk} corresponds to
  a qubit channel if and only if $Q_{11},Q_{22}\in\M{2}$, 
  $Q_{21} = Q_{12}^\ast$, $\Tr Q_{12}=0$ and $Q\ge 0$ which means
  that the space of qubit channels with real and complex entries can be  
  identified with convex subsets of $\R^7$ and $\R^{12}$, respectively.
We introduce the following notations for these sets.
\begin{align*}
 & \Q_\R = \{Q\in \R^{4\times 4}|Q:\MR{2}\to\MR{2},Q> 0\}\\
 & \Q_\C = \{Q\in \C^{4\times 4}|Q:\MC{2}\to\MC{2},Q> 0\}
\end{align*}
A general element can be parametrized as
\begin{equation}\label{eq:matQ}
 Q = \left(
 \begin{array}{cccc}
  a &   b & c &  d \\
  \bar{b} & 1-a & e & -c \\
  \bar{c} &   \bar{e} & f &  g \\
  \bar{d} &  -\bar{c} & \bar{g} & 1-f
 \end{array}
 \right),
\end{equation}
  where $a,f\in [0,1]$ and $Q>0$.
Let us choose the unitary matrix
\begin{equation}\label{eq:unit}
  U = \left(
 \begin{array}{cccc}
  1 & 0 & 0 & 0 \\
  0 & 0 & 1 & 0 \\
  0 & 1 & 0 & 0 \\
  0 & 0 & 0 & 1
 \end{array}
 \right)
\end{equation}
  and define the matrix $A$ as
\begin{equation}
 A = U^\ast Q U = \left(
 \begin{array}{cccc}
  a &   c &   b &  d \\
  \bar{c} &   f &   e &  g \\
  \bar{b} &   \bar{e} & 1-a & -c \\
  \bar{d} &   \bar{g} &  -\bar{c} & 1-f
 \end{array}\label{eq:matA}
 \right)
\end{equation}
  which is positive definite if and only if $Q$ is positive definite
  hence $A$ gives an equivalent parametrization of $\Q_\R$ and $\Q_\C$.

\begin{lemma}\label{lem:proj}
Let $A$ be an $n\times n$ positive definite matrix, $T = \det (A) A^{-1}$,
$L\subseteq\K^n$ a subspace, $x\in L^\perp$ and $M = \sqrt{T}L$. 
If $\dim (L^\perp) = 1$, then
\begin{equation*}
||P_{M^\perp}\sqrt{T}x||^2=\frac{\det (A)}{\left<x,Ax\right>}||x||^4.
\end{equation*}
\end{lemma}

\begin{proof}
According to Remark \ref{rem:1} $M^\perp = \sqrt{A}(L^\perp)$.
 If $\dim (L^\perp) = 1$, then $\{b_1=||\sqrt{A}x||^{-1}\sqrt{A}x\}$
 is an orthonormal basis of $M^\perp$ hence $P_{M^\perp}=b_1\otimes b_1$.
We can write 
 \begin{equation*}
  ||P_{M^\perp}\sqrt{T}x||^2 = \det (A) \left|\left<b_1,\sqrt{A^{-1}}x\right>\right|^2
  = \frac{\det (A)}{\left<x,Ax\right>}||x||^4
 \end{equation*}
  which completes the proof.
\end{proof}

% Reals
\begin{theorem}
The volume of the space $\Q_\R$ with respect to the Lebesgue measure is
 \begin{equation*}
  V(\Q_\R) = \frac{4\pi^3}{105},
 \end{equation*}
  and the distribution of volume over classical channels can be written as
 \begin{equation*}
 V(a,f) = \frac{128}{45}\pi^{2}
 \times
 \begin{cases}
  (a f)^{3/2} \left(5(1-a)(1-f)-af\right)
  & \text{if } a+f<1\\
  ((1 - a) (1 - f))^{3/2} (5 af-(1-a)(1-f))
  & \text{if } a+f\ge 1.
 \end{cases}
 \end{equation*}
\end{theorem}
\begin{proof}
The volume element corresponding to the parametrization \eqref{eq:matQ}
  in the real case is $2^4\dint \la_7$. A matrix of the form \eqref{eq:matA} 
  with real entries represents a point of $\Q_\R$ if and only if
  $a,f\in [0,1]$ and $\det (A_i) > 0$ for $i = 1,2,3,4$. 
First we assume that $a$ and $f$ are given. 

If $A_3$ is fixed, then by Lemma \ref{lem:integral} and Lemma \ref{lem:proj}
  we have
\begin{align*}
 V(A_3) &= \int\limits_{E^\R (T_3,(1-f)\det (A_3),L_3,x_3)} 2^4 \dint \la_2
 \\
 &=\frac{2^4 F_{1}G_{1,0}\left((1-f)-\frac{c^2}{1-a}\right)_{+}\det (A_3)}{\sqrt{\det (T_3|_{L_3})}}
 \\
 &=
 2^4 F_1 G_{1,0}
 \frac{\left(
 (1-a)(1-f) - c^2 
 \right)_{+}}{(1-a)^{3/2}}
 \det (A_3)^{1/2},
\end{align*}
  where $L_3 = \Span \{(1,0,0)^T,(0,1,0)^T\}$ and $x_3 = (0,0,-c)^T$.

If $A_2$ is fixed, then
\begin{align*}
 V(A_2) &= \int\limits_{\E^\R (T_2,(1-a)\det (A_2),\R^2,0)}
 V(A_3)
 \dint \la_2 \\
 &=
 2^4 F_1 G_{1,0}
 \frac{\left(
 (1-a)(1-f) - c^2 
 \right)_{+}}{(1-a)^{3/2}}
 \\
 &\times
 \int\limits_{\E^\R (T_2,(1-a)\det (A_2),\R^2,0)}
 \left((1-a)\det (A_2)-\langle y,T_2 y\rangle\right)^{1/2}
 \dint \la_2 (y)\\
 &=
 2^4 F_1^2 G_{1,0} G_{1,\frac{1}{2}}
 \left(
 (1-a)(1-f) - c^2 
 \right)_{+} \det (A_2).
\end{align*}
Observe that $af-c^2>0$ implies $(1-a)(1-f)-c^2>0$ whenever $a+f\le 1$ 
  and $(1-a)(1-f)-c^2>0$ implies $af-c^2>0$ if $a+f\ge 1$ holds.
Since $\frac{2^{6}}{15}F_{1}^{2}G_{1,0}G_{1,\frac{1}{2}}=\frac{128}{45}\pi^{2}$
  the volume element corresponding to a fixed $a$ and $f$ can be expressed as
\begin{equation*}
 V(a,f) = \frac{128}{45}\pi^{2}
 \times
 \begin{cases}
  (a f)^{3/2} \left(5(1-a)(1-f)-af\right)
  & \text{if } a+f<1\\
  ((1 - a) (1 - f))^{3/2} (5 af-(1-a)(1-f))
  & \text{if } a+f\ge 1
 \end{cases}
\end{equation*}
  (see Figure \ref{Fig:vafR}) 
\begin{figure}[!ht]
 \centering
  \includegraphics[width = 0.75\textwidth]{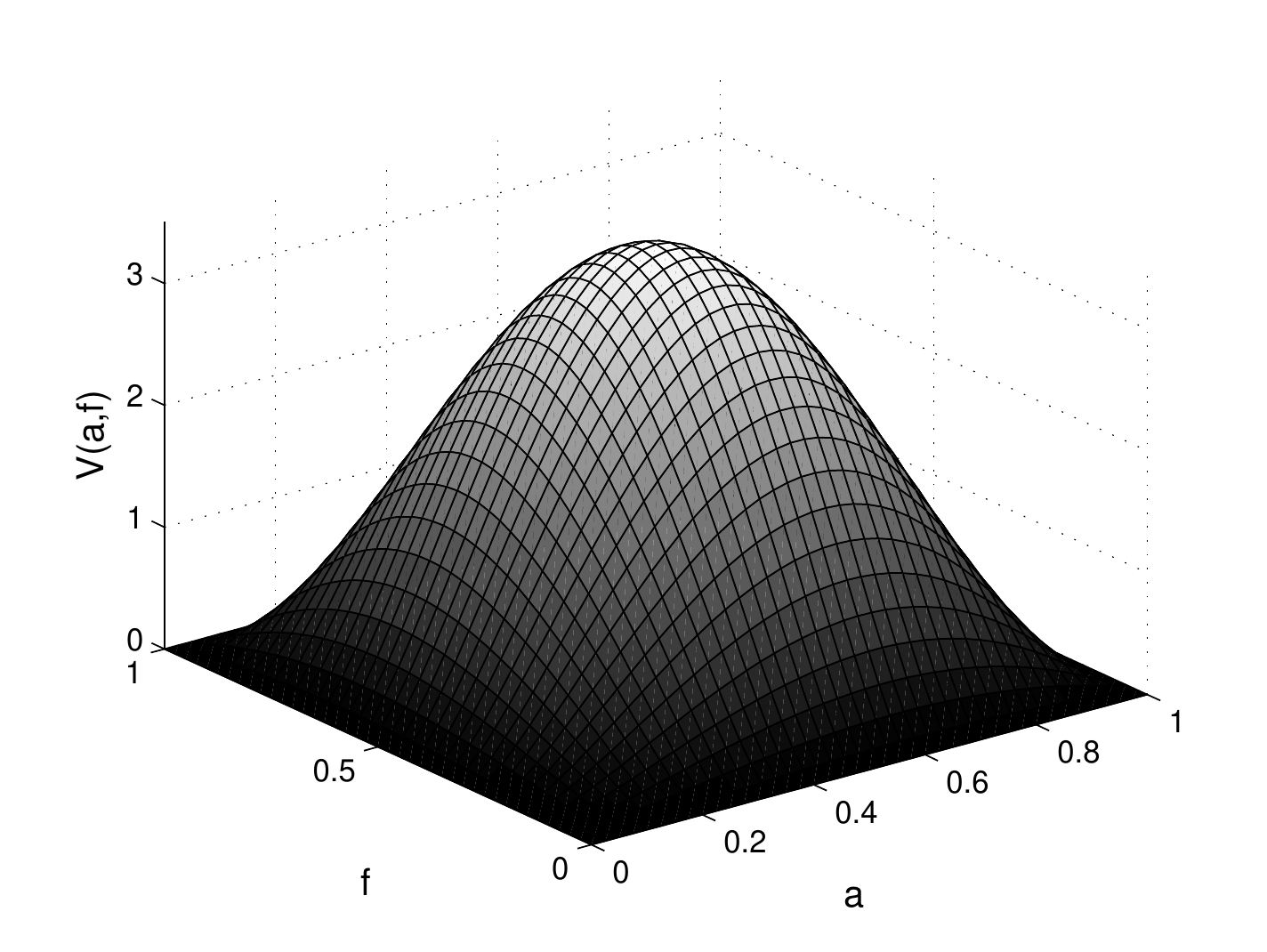}
  \caption{Graph of $V(a,f)$ for $\Q_\R$.}
  \label{Fig:vafR}  
\end{figure} 
  thus for the volume of $\Q_\R$ we have 
\begin{equation*}
 V(\Q_\R)=\int\limits_{[0,1]^2} V(a,f)\dint a\dint f
 = \frac{4\pi^3}{105}\approx 1.18119
\end{equation*}
  which completes the proof.
\end{proof}

% Complex
\begin{theorem}
The volume of the space $\Q_\C$ with respect to the Lebesgue measure is
 \begin{equation*}
  V(\Q_\C)=\frac{2 \pi ^5}{4725},
 \end{equation*}
  and the distribution of volume over classical channels can be written as
 \begin{align*}
 V(a,f) &= \frac{16}{45}\pi^{5}\\
 &\times
 \begin{cases}
  a^3 f^3 [
  10((1-a)(1-f)-af)^2 +
  & \\
  15 af(1-a)(1-f)-9a^2f^2
  ]
  & \text{if } a+f< 1\\
 (1-a)^3 (1-f)^3 [
 10((1-a)(1-f)-af)^2 +
  & \\
  15 af(1-a)(1-f)-9(1-a)^2(1-f)^2
 ]
  & \text{if } a+f\ge 1.
 \end{cases}
\end{align*}
\end{theorem}
\begin{proof}
The volume element corresponding to the parametrization \eqref{eq:matQ}
  in the complex case is $2^7 \dint \la_{12}$.
Similar to the real case, a matrix of the form \eqref{eq:matA} with complex 
  entries represents a point of $\Q_\C$ if and only if $a,f\in [0,1]$ and
 $\det (A_i) > 0$ for $i=1,2,3,4$. First we assume that $a$ and $f$ are given.
 
If $A_3$ is fixed, then by Lemma \ref{lem:integral} and Lemma \ref{lem:proj}
  we have
\begin{align*}
 V(A_3) &= \int\limits_{\E^\C (T_3,(1-f)\det (A_3),L_3,x_3)} 2^7 \dint \la_4
 \\
 &=\frac{2^7 F_{3}G_{3,0}\left((1-f)-\frac{|c|^2}{(1-a)}\right)_{+}^2 \det (A_3)^2}{\det (T_3|_{L_3})}
 \\
 &=
 \frac{2^7 F_{3}G_{3,0}}{(1-a)^3}
 \left((1-a)(1-f) - |c|^2\right)_{+}^2 
 \det (A_3),
\end{align*}
  where $L_3 = \Span \{(1,0,0)^T,(0,1,0)^T\}$ and $x_3 = (0,0,-c)^T$.

If $A_2$ is fixed, then
\begin{align*}
 & V(A_2) = \int\limits_{\E^\C (T_2,(1-a)\det (A_2),\C^2,0)}
 V(A_3)
 \dint \la_4 =\\
 &=
 \frac{2^7 F_{3}G_{3,0}}{(1-a)^3}
 \left((1-a)(1-f) - |c|^2\right)_{+}^2
 \\
 &\times
 \int\limits_{E^\C (T_2,(1-a)\det (A_2),\C^2,0)}
(1-a)\det (A_2)-\langle y,T_2 y\rangle
 \dint \la_4 (y)\\
 &= 
 2^7 F_3^2 G_{3,0} G_{3,1}
 \left((1-a)(1-f)-|c|^2\right)_{+}^2 \det (A_2)^2.
\end{align*}
 
The volume corresponding to a fixed $a$ and $f$ can be expressed as
\begin{align*}
 V(a,f) &= \frac{2^7}{60} F_1 F_3^2 G_{3,0} G_{3,1} \\
 &\times
 \begin{cases}
  a^3 f^3 [
  10((1-a)(1-f)-af)^2 +
  & \\
  15 af(1-a)(1-f)-9a^2f^2
  ]
  & \text{if } a+f< 1\\
 (1-a)^3 (1-f)^3 [
 10((1-a)(1-f)-af)^2 +
  & \\
  15 af(1-a)(1-f)-9(1-a)^2(1-f)^2
 ]
  & \text{if } a+f\ge 1
 \end{cases}
\end{align*}
  (see Figure \ref{Fig:vafC}) thus for the volume of $\Q_\C$ we have
\begin{equation*}
 V(\Q_\C) = \int\limits_{[0,1]^2} V(a,f)\dint a\dint f
 = \frac{2 \pi ^5}{4725} \approx 0.129532.
\end{equation*}
 \begin{figure}[!ht]
 \centering
  \includegraphics[width = 0.75\textwidth]{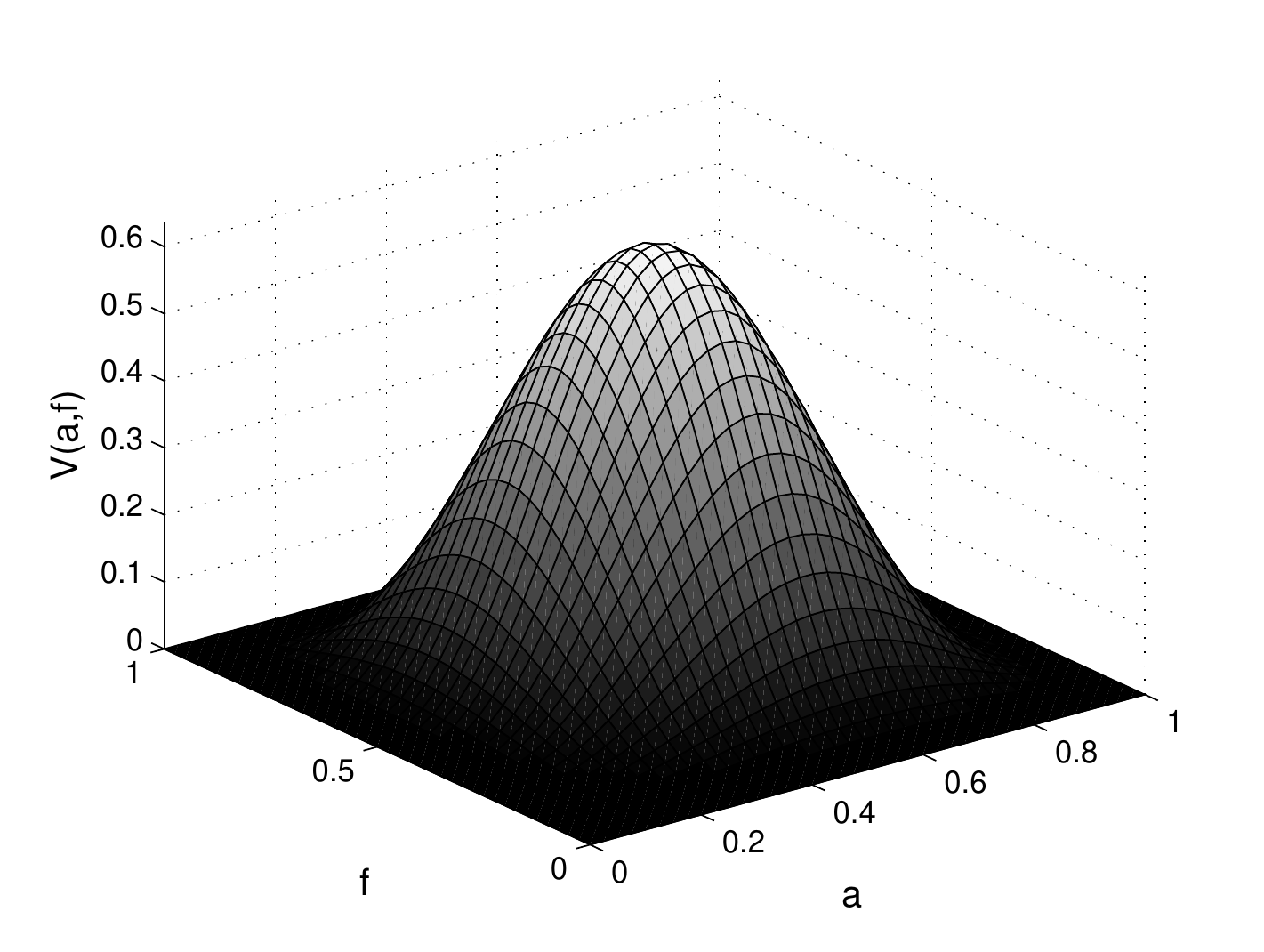}
  \caption{Graph of $V(a,f)$ for $\Q_\C$.}
  \label{Fig:vafC}  
\end{figure}
\end{proof}

\subsection{Unital qubit channels}

Identity preserving requires that $Q_{11}=Q_{22}=I$ in the Choi representation
  \eqref{eq:blk} which means that the space of unital qubit channels 
  with real and complex entries can be identified with convex subsets of
  $\R^5$ and $\R^9$, respectively.
We introduce the following notations for these sets.
\begin{align*}
 & \Q_\R^1 = \{Q\in \R^{4\times 4}|Q:\MR{2}\to\MR{2},Q> 0,Q(I)=I\}\\
 & \Q_\C^1 = \{Q\in \C^{4\times 4}|Q:\MC{2}\to\MC{2},Q> 0,Q(I)=I\}
\end{align*}
A general element can be parametrized as
\begin{equation}\label{eq:matQ2}
 Q = \left(
 \begin{array}{cccc}
  a &   b &   c &  d \\
  \bar{b} & 1-a &   e & -c \\
  \bar{c} &   \bar{e} & 1-a & -b \\
  \bar{d} &  -\bar{c} &  -\bar{b} &  a
 \end{array}
 \right),
\end{equation}
  where $a\in [0,1]$ and $Q>0$. 
Let us choose the unitary matrix
\begin{equation*}
  U = \left(
 \begin{array}{cccc}
  0 & 0 & 1 & 0 \\
  1 & 0 & 0 & 0 \\
  0 & 1 & 0 & 0 \\
  0 & 0 & 0 & 1
 \end{array}
 \right)
\end{equation*}
  and define the matrix $A$ as
\begin{equation}\label{eq:matA2}
 A = U^\ast Q U = \left(
 \begin{array}{cccc}
1-a &   e &   b & -c \\
  \bar{e} & 1-a &   c & -b \\
  \bar{b} &   \bar{c} &   a &  d \\
 -\bar{c} &  -\bar{b} &   \bar{d} &  a
 \end{array}
 \right)
\end{equation}
  which is positive definite if and only if $Q$ is positive definite.

\begin{lemma}\label{lem:complex}
Let us denote the left upper $k\times k$ submatrix of $A$ by $A_k$.
If $L_3 = \Span \{(0,0,1)^T\}$ and $M=\sqrt{A_3^{-1}}(L_3)$, then
$
 \sqrt{A_3^{-1}}P_{M^\perp}\sqrt{A_3^{-1}}
 = \left(
 \begin{array}{cc}
  A_2^{-1} & \mathbf{0} \\
  \mathbf{0}^T & 0
 \end{array}
 \right)$.
\end{lemma}
\begin{proof}
According to Remark \ref{rem:1} $M^\perp = \sqrt{A}(L^\perp)$.
If $u_1$ and $u_2$ are vectors in $L_3^\perp$ for which
 $\left<u_i,A_3 u_j\right> = \delta_{ij}$ holds, then 
 $\{\sqrt{A_3}u_1,\sqrt{A_3}u_2\}$ is an orthonormal basis of $M^\perp$
 hence 
 \begin{equation*}
 P_{M^\perp}=\sqrt{A_3}u_1 \otimes\sqrt{A_3}u_1 
  +\sqrt{A_3}u_2\otimes\sqrt{A_3}u_2
 = \sqrt{A_3} (u_1\otimes u_1 + u_2\otimes u_2) \sqrt{A_3}
 \end{equation*}
  which implies that 
$\sqrt{A_3^{-1}}P_{M^\perp}\sqrt{A_3^{-1}} 
  = u_1\otimes u_1 + u_2\otimes u_2$.
Let us define the matrix
$B = 
  \left(
  \begin{array}{cc}
   A_2 & \mathbf{0} \\
   \mathbf{0}^T & 1
  \end{array}
  \right).$
It is easy to see that $\left<x,By\right>=\left<x,A_3y\right>$ holds for each
  $x,y\in L_3^\perp$.
We can choose $u_i=\sqrt{B^{-1}}e_i$, $i=1,2$, where $(e_i)_j = \delta_{ij}$, 
  $i,j=1,2$ is the standard basis of $L_3^\perp$.
We can write
$u_1\otimes u_1 + u_2\otimes u_2 
=\sqrt{B^{-1}}(e_1\otimes e_1 + e_2\otimes e_2)\sqrt{B^{-1}}
=\left(\begin{array}{cc}
  A_2^{-1} & \mathbf{0} \\
  \mathbf{0}^T & 0
 \end{array}\right)$ which completes the proof.
\end{proof}

% Reals
\begin{theorem}
The volume of the space $\Q_\R^1$ with respect to the Lebesgue measure is
 \begin{equation*}
  V(\Q_\R^1) = \frac{4\pi^2}{15},
 \end{equation*}
  and the distribution of volume over classical channels can be written as 
 \begin{equation*}
  V(a) = 8\pi^{2} a^2 (1-a)^2.
 \end{equation*}
\end{theorem}
\begin{proof}
The volume element corresponding to the parametrization
  \eqref{eq:matQ2} in the real case is  $2^4\dint\la_5$.
A matrix of the form \eqref{eq:matA2} with real entries represents a point 
  of $\Q_\R^2$ if and only if $a\in [0,1]$ and $\det (A_i)>0$ for $i=1,2,3,4$. 
First we assume that $a$ is given.
 
If $A_3$ is fixed, then by Lemma \ref{lem:integral} and 
  Lemma \ref{lem:complex}, we have
 \begin{align*}
  V(A_3) &= \int\limits_{\E^\R (T_3,a\det (A_3),L_3,x_3)} 2^4\dint\la_1 =\\
  &=\frac{2^4 F_0}{\sqrt{\det (A_2)}}
  \left(a-\left< x_3,
 \left(
 \begin{array}{cc}
  A_2^{-1} & \mathbf{0} \\
  \mathbf{0}^T & 0
 \end{array}
 \right) x_3\right>\right)_{+}^{1/2}\det (A_3)^{1/2},
 \end{align*}
  where $L_3 = \Span \{(0,0,1)^T\}$ and $x_3 = (-c,-b,0)^T$.
 
 Observe that $\left< x_3,
 \left(
 \begin{array}{cc}
  A_2^{-1} & \mathbf{0} \\
  \mathbf{0}^T & 0
 \end{array}
 \right) x_3\right>
 =
 \left<y,\sigma_1 A_2^{-1}\sigma_1 y\right> =
 \left<y, A_2^{-1} y\right>
 $, where $y = (b,c)^T$ because 
 $\sigma_1 A_2^{-1}\sigma_1 = \overline{A_2^{-1}}$
   and $A_2^{-1}$ is a matrix with real entries.
 
 If $A_2$ is fixed, then  
 \begin{align*}
  V(A_2) &= \int\limits_{\E^\R (T_2,a\det (A_2),\R^2,0)}
  V(A_3) \dint\la_2 \\
  & =  
  \frac{2^4 F_0 }{\sqrt{\det (A_2)}}
  \int\limits_{\E^\R (T_2,a\det (A_2),\R^2,0)}
  \left(a-
  \left< 
  y, A_2^{-1} y
  \right>
 \right)_{+}^{1/2}\det (A_3)^{1/2}
  \dint\la_2 (y)\\
  & =
  \frac{2^4 F_0 }{\det (A_2)}
  \int\limits_{\E^\R (T_2,a\det (A_2),\R^2,0)}
  a\det (A_2)-
  \left< 
  y,  T_2y
  \right>
 \dint\la_2 (y) \\
 & = 2^4F_0F_1G_{1,1}a^2\left(\det (A_2)\right)^{1/2}.
 \end{align*}
 
The volume corresponding to a fixed $a\in [0,1]$ can be written as
 \begin{align*}
  V(a) &=
  2^4F_0F_1G_{1,1}a^2
  \int\limits_{-(1-a)}^{1-a}
  \sqrt{(1-a)^2-e^2}
  \dint\la_1(e) \\
  &= 2^5 F_0F_1 G_{0,\frac{1}{2}} G_{1,1} a^2 (1-a)^2=8\pi^{2}a^2 (1-a)^2
 \end{align*}
  (see Figure \ref{Fig:unital}) thus the volume of $\Q_\R^1$ is
\begin{equation*}
 V(\Q_\R^1) =
 8\pi^{2} \int\limits_0^1 a^2 (1-a)^2  \dint a=\frac{4\pi^2}{15}
  \approx 2.63189
\end{equation*}
  which completes the proof.
\end{proof}

% Complex
\begin{theorem}
The volume of the space $\Q_\C^1$ with respect to the Lebesgue measure is
 \begin{align*}
  V(\Q_\C^1) &=\frac{2\pi^4}{315},
 \end{align*}
  and the distribution of volume over classical channels can be written as 
 \begin{equation*}
  V(a) = 2^2\pi^4 a^4 (1-a)^4.
 \end{equation*}
\end{theorem}
\begin{proof}
The volume element corresponding to the parametrization \eqref{eq:matQ2} 
  in the complex case is $2^7\dint\la_9$.
Similar to the real case, a matrix of the form \eqref{eq:matA2} with complex 
  entries  represents a point of $\Q_\R^2$ if and only if $a\in [0,1]$
  and $\det (A_i)>0$ for $i=1,2,3,4$. 
First we assume that $a$ is given.
 
If $A_3$ is fixed, then by Lemma \ref{lem:integral} and 
  Lemma \ref{lem:complex}, we have
\begin{align*}
  V(A_3) &= \int\limits_{\E^\C (T_3,a\det (A_3),L_3,x_3)} 2^7\dint\la_2 =\\
  &=\frac{2^6 F_1}{\det(A_2)} \left(a-\left< x_3,
  \left(\begin{array}{cc}
  A_2^{-1} & \mathbf{0} \\
  \mathbf{0}^T & 0
 \end{array}\right) x_3\right>\right)_{+}\det (A_3),
\end{align*}
  where $L_3 = \Span \{(0,0,1)^T\}$ and $x_3 = (-c,-b,0)^T$.

 Similar to the real case $\left< x_3,
 \left(
 \begin{array}{cc}
  A_2^{-1} & \mathbf{0} \\
  \mathbf{0}^T & 0
 \end{array}
 \right) x_3\right>
 =\left<y,\sigma_1 A_2^{-1}\sigma_1 y\right>$,
   where $y = (b,c)^T$, but
 $\left<y,\sigma_1 A_2^{-1}\sigma_1 y\right>
 \ne \left<y, A_2^{-1} y\right>$
   because $A_2^{-1}\ne\overline{A_2^{-1}}$ in the complex case.
 
If $A_2$ is fixed, then
 \begin{align*}
  &V(A_2) = \int\limits_{\E^\C (T_2,a\det (A_2),\C^2,0)}
  V(A_3) \dint\la_4 =\\
  &= 2^6 F_1
  \int\limits_{\E^\C (T_2,a\det (A_2),\C^2,0)}
  \left(
  a-\left<y,\sigma_1 A_2^{-1}\sigma_1 y\right>
  \right)_{+}(a-\left<y,A_2^{-1} y\right>)
 \dint\la_4(y).
 \end{align*}
Let us substitute $y = \sqrt{a}\sqrt{A_2}z$ and obtain
 \begin{equation*}
 V(A_2)=2^6 F_1 a^4 \det(A_2)
 \int\limits_{\{z:||z||<1\}}
 \left(1-\left<z,Bz\right>\right)_{+}
 (1-||z||^2) \dint\la_4(z),
 \end{equation*}
  where $B=\sqrt{A_2}\sigma_1 A_2^{-1}\sigma_1\sqrt{A_2}$ is a self-adjoint 
  matrix that is unitary equivalent to a diagonal matrix and $\det (B)=1$.
As a unitary coordinate transformation does not change the value of 
  the previous integral hence
\begin{align*}
 V(A_2) &= 2^6 F_1 a^4 \det(A_2)\\
 & \times
 \int\limits_{\{z:||z||<1\}}
 \left(1-\la |z_1|^2-\frac{1}{\la} |z_2|^2\right)_{+}
 (1-|z_1|^2-|z_2|^2) \dint\la_4(z),
\end{align*}
  where $\la$ denotes the largest eigenvalue of $B$.
Then we compute the integral above in the Descartes product of two polar   
  coordinate systems.
The integral with respect to the angles gives $F_1^2$ and the radial part
  can be written as
\begin{align*}
 V(A_2) &= 2^6 F_1^3 a^4 \det(A_2)
 \int\limits_{\R_{+}^2}
 \left(1-\la r_1^2-\frac{1}{\la} r_2^2\right)_{+}
 (1-r_1^2-r_2^2)_{+} r_1 r_2\dint r_1 \dint r_2 \\
 &= \frac{2^5\pi^3}{3}a^4 \det(A_2) 
 \frac{3\la-1}{\la (1+\la)}.
\end{align*}
By elementary matrix computation, we get
\begin{equation*}
\la = 1+\frac{2\Im(e)^2}{\det (A_2)}
+ \sqrt{\left(1+\frac{2\Im(e)^2}{\det (A_2)}\right)^2-1}
\end{equation*}
  thus
\begin{equation*}
 V(A_2) =
 \frac{2^5\pi^3}{3} a^4 \det(A_2) 
 \left(
 1+\frac{2\Im (e)^2}{\det (A_2)}
 \left( 
 \sqrt{\frac{\Im (e)^2}{\det (A_2)+\Im (e)^2}}-1
 \right)
 \right).
\end{equation*}

The volume corresponding to a fixed $a\in [0,1]$ can be written as
 \begin{align*}
  V(a) &=
  \int\limits_{|e|^2\le (1-a)^2}
  V(A_2)
  \dint\la_2(e) \\
  &=
  \frac{2^5\pi^3}{3} a^4 (1-a)^4 \\
  &\times
  \int\limits_0^1
  \int\limits_0^{2\pi}
  \left(
  1+
  \frac{2r^2\sin^2\phi}{1-r^2}
  \left(
  \sqrt{\frac{r^2\sin^2\phi}{1-r^2\cos^2\phi}}-1
  \right)
  \right)
  (1-r^2)r
  \dint \phi
  \dint r \\
  & = 2^2\pi^4 a^4 (1-a)^4
 \end{align*} 
(see Figure \ref{Fig:unital}) thus the volume of $\Q_\C^1$ is
\begin{equation*}
 V(\Q_\C^1) =
 2^2\pi^4 
 \int\limits_0^1
 a^4 (1-a)^4
 \dint a
 = \frac{2\pi^4}{315}
 \approx 0.61847
\end{equation*}
  which completes the proof.
\end{proof}
\begin{figure}[!ht]
\centering
 \includegraphics[width = 0.7\textwidth]{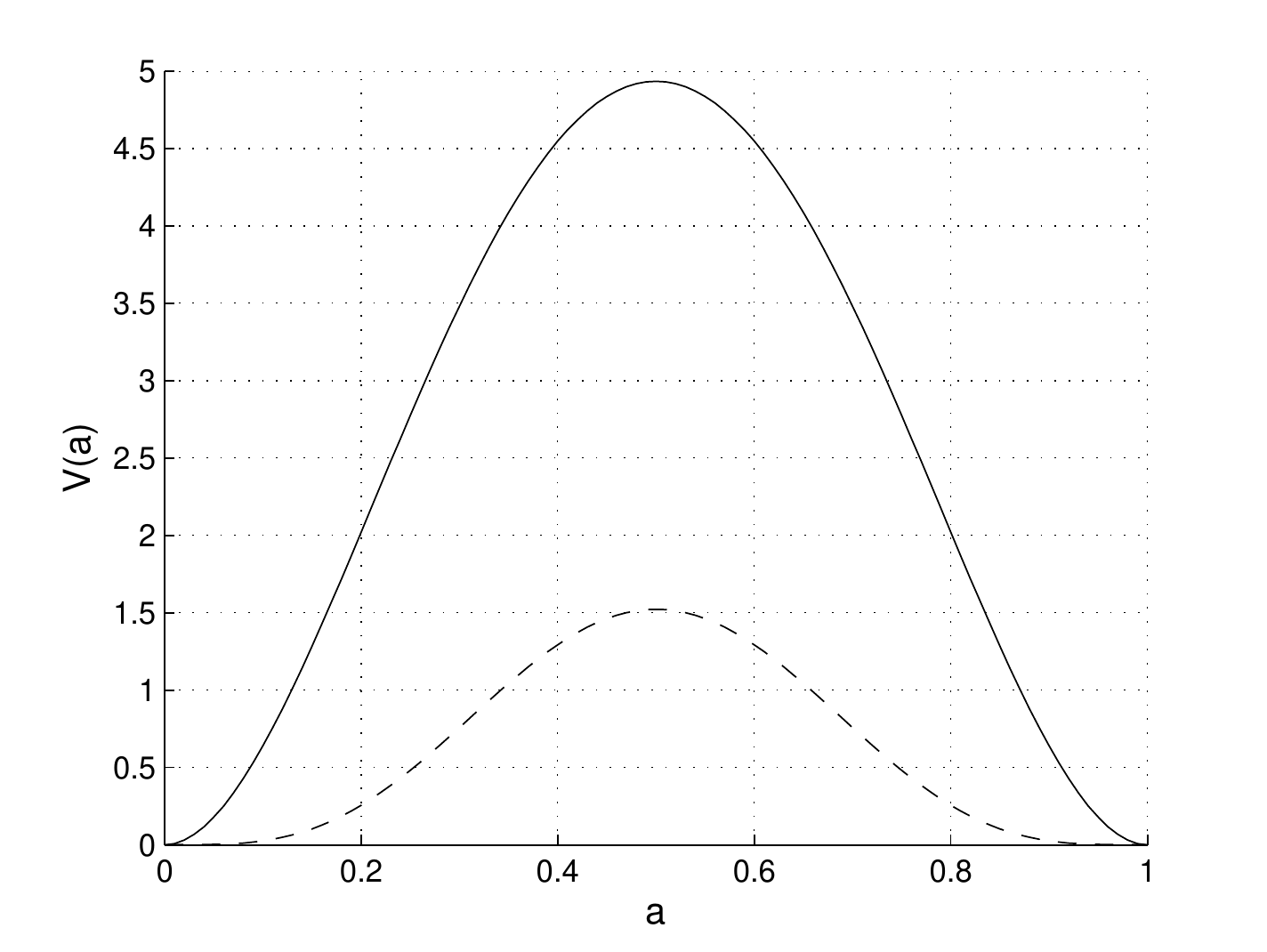}
\caption{Graph of $V(a)$ 
for $\Q_\R^1$ (solid) 
and 
$\Q_\C^1$ (dashed).} 
\label{Fig:unital}
\end{figure}

One might think about the generalization of the presented results,
  although in a more general setting several complications occur.
For example, in the case of unital qubit channels one should integrate
  over the Birkhoff polytope, which would cause difficulties since
  even the volume of the polytope is still unknown \cite{IgorPak}.

\section{The trace-norm contraction coefficient}

The way of integration presented in the previous sections suggests an efficient
  method for generating uniformly distributed points in the space of qubit 
  channels.
This method makes the numerical study of different channel related quantities
  possible.
As an example, the distribution of $\eta^{\Tr}$ is investigated numerically by
  Monte-Carlo simulations over different kind of quantum channels.

\subsection{Monte-Carlo simulations}

Simulations were implemented in MATLAB 2014a and random vectors within a 
  sphere were generated, as described by Knuth \cite{Knuth}.

\begin{algorithm}\label{alg:alg}
The next scheme describes for the case of $\Q_\R$ how the algorithm works, 
  where $x\sim\mathcal{U}(B)$ denotes that $x$ is uniformly distributed on the
  set $B$.
The other cases ($\Q_\C$, $\Q_\R^1$ and $\Q_\C^1$) can be treated in a 
  similar way.
\begin{center}
 \begin{tabular}{ll}
  \hline\hline \vspace{-6pt}\\
  \textbf{Step 1:} &  Generate $a,f\sim\mathcal{U}([0,1])$ independently. \medskip\\
  \textbf{Step 2:} &  Generate $x_1 \sim\mathcal{U}(-\sqrt{af},\sqrt{af})$ and set 
                          $A_2 = \left(\begin{array}{cc}
					  a & x_1 \\
					  x_1 & f
                                        \end{array}
				  \right)$. \medskip\\
  \textbf{Step 3:} &  Generate $y_2\sim\mathcal{U}(\{r\in\R^2:||r||\le\sqrt{f}\})$ and \\
                   &  set $A_3 = \left(
                   \begin{array}{cc}
                    A_2 & x_2 \\
                    x_2^T & 1-a
                   \end{array}
                   \right)
                   $, where $x_2=\sqrt{A_2}y_2$.\medskip\\
  \textbf{Step 4:} & Compute the projection $P$ onto the subspace
  $\Span (\{\sqrt{A_3}\,e_3\})$\\ 
                   & and set $z=-x_1 (1-f)^{-1/2}P\sqrt{A_3^{-1}}\,e_3$.\medskip\\                
  \textbf{Step 5:} &  \textbf{If} $||z|| > 1$, \textbf{then} 
                      \textbf{goto} Step 2. \medskip\\
  \textbf{Step 6:} & Generate $y_3\sim\mathcal{U}(\{r\in\R^2:||r||\le\sqrt{1-||z||^2}\})$ and \\
                   & set $A=\left(
                   \begin{array}{cc}
                    A_3 & x_3 \\
                    x_3^T & 1-f
                   \end{array}
                   \right)$, where $x_3=\sqrt{1-f}\sqrt{A_3}([e_1,e_2]y_3+z)$.\medskip\\                   
  \textbf{Step 7:} & Apply the transform $Q=UAU^\ast$, where $U$ is given by
  \eqref{eq:unit}. \medskip\\
  \hline
 \end{tabular}
\end{center} 
The first step is omitted when our goal is to generate a random qubit channel 
  over the classical channel parametrized by $a$ and $f$. 
Step 5 is needed just because  up to this point it was not guaranteed that
  $c^2\le (1-a)(1-f)$.
\end{algorithm}

Any $\rho\in\M{2}$ can be represented in the Pauli bases as
  $\rho = \frac{1}{2}(I+x\cdot{\sigma})$ by a unique 
  $x=(x_1,x_2,x_3)^T\in\R^3$ with $||x||\le 1$, where 
  $x\cdot {\sigma} = \sum\limits_{j=1}^3 x_i \sigma_i$.
A qubit channel $Q:\M{2}\to\M{2}$ is represented in the Pauli bases as
\begin{equation*}
 Q\left(\frac{1}{2}(I+x\cdot\sigma)\right) = \frac{1}{2}\left(I+(v+Tx)\cdot{\sigma}\right),
\end{equation*}
  where $v\in\R^3$ and $T$ is a $3\times 3$ real matrix.
This representation is suitable for calculating the trace-norm 
  contraction coefficient because $\eta^{\Tr}(Q)$ can be expressed
  as $\eta^{\Tr}(Q)=||T||_\infty$, where $||.||_\infty$
  denotes the Schatten-$\infty$ norm \cite{NielsenChuang}. 
It means that the trace distance contraction coefficient of a 
  qubit channel $Q$ given by \eqref{eq:matQ} is the largest singular value
  of the following matrix. 
\begin{equation}\label{eq:matT}
 T=
 \left(
 \begin{array}{ccc}
  \Re (d+e) & \Im (d+e) & \Re (b-g)  \\
 -\Im (d-e) & \Re(d-e)  & -\Im (b-g) \\
   2\Re (c) & 2\Im (c)  & a-f
 \end{array}
 \right)
\end{equation}

\subsection{Distribution of $\eta^{\Tr}$ on the whole space}

Empirical cumulative distribution functions (CDF) of
  $\eta^{\Tr}$ on the space of qubit channels are presented 
  in Figure \ref{Fig:F1234} for $\Q_\R,\Q_\C,\Q_\R^1$ and $\Q_\C^1$. 
In each case, $10^4$ random qubit channels were generated independently
  and confidence band corresponding to the confidence level 
  $99.995\%$ ($\alpha = 5\times 10^{-5}$) was calculated by Greenwood's 
  formula \cite{Lawless}.
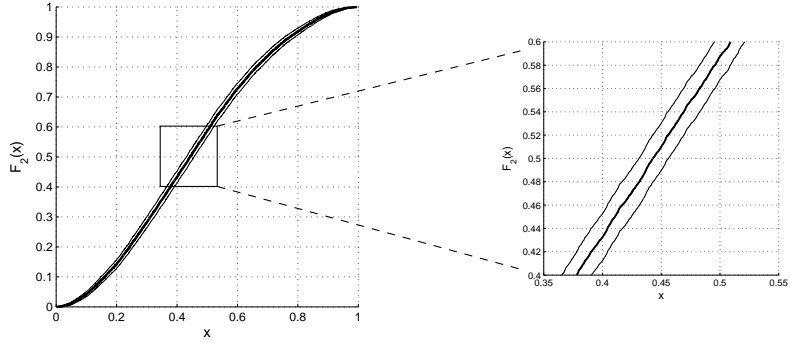
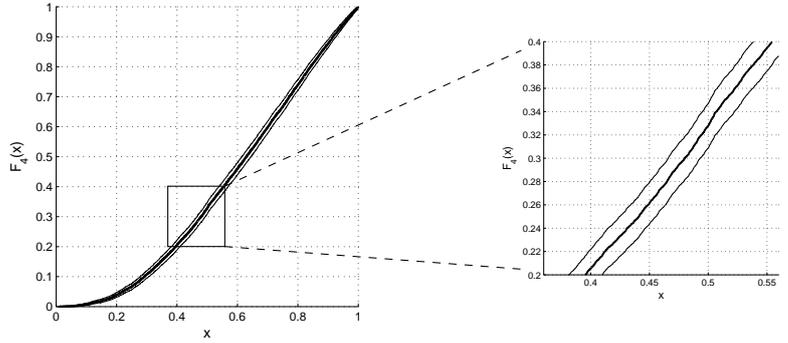
\begin{figure}[!ht]
\centering
\begin{subfigure}{\textwidth}
\centering
 \begin{tikzpicture}
  \draw (0,0) node {\includegraphics[scale = 0.45]{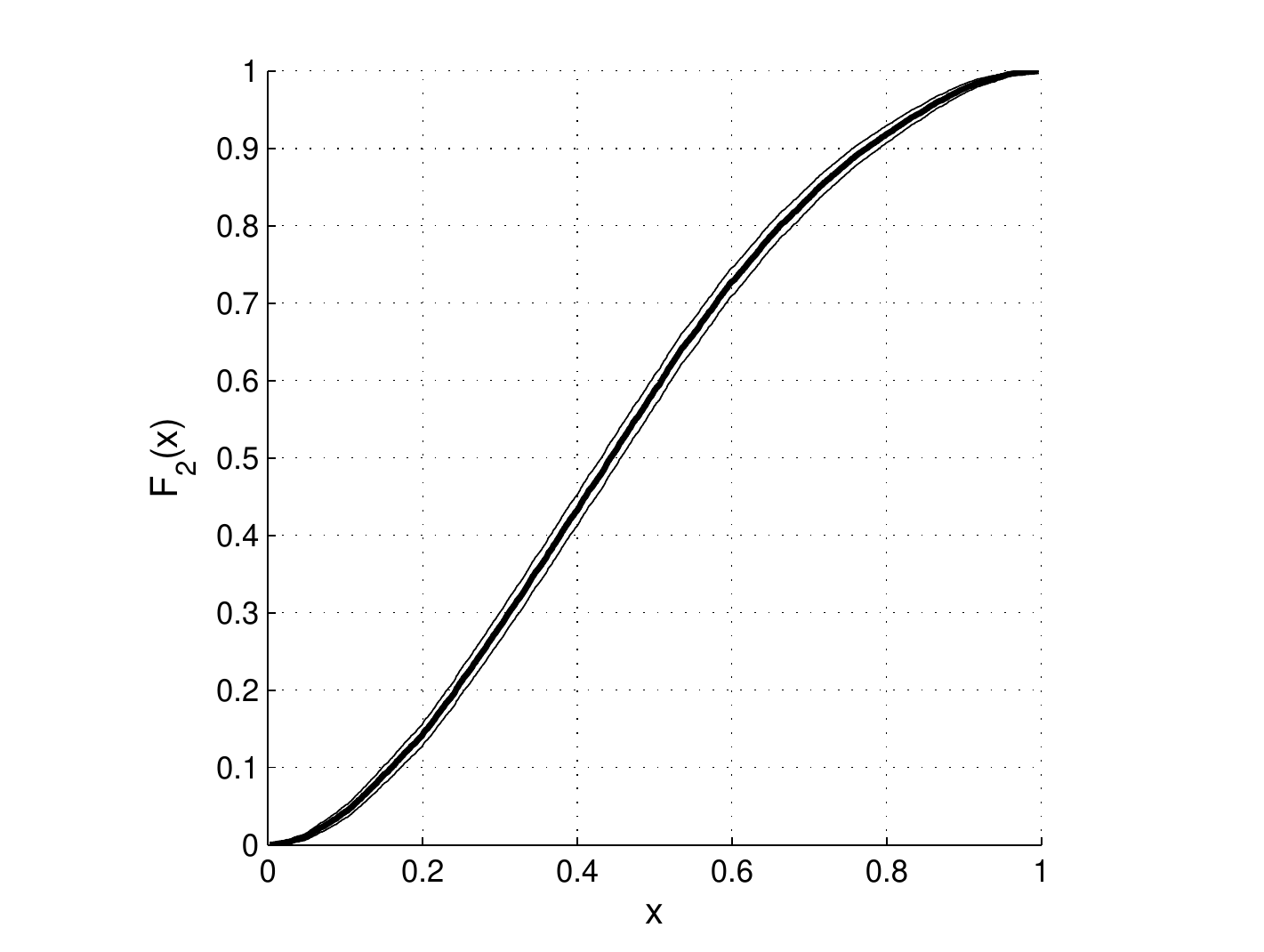}};
  \draw (6,0) node {\includegraphics[scale = 0.35]{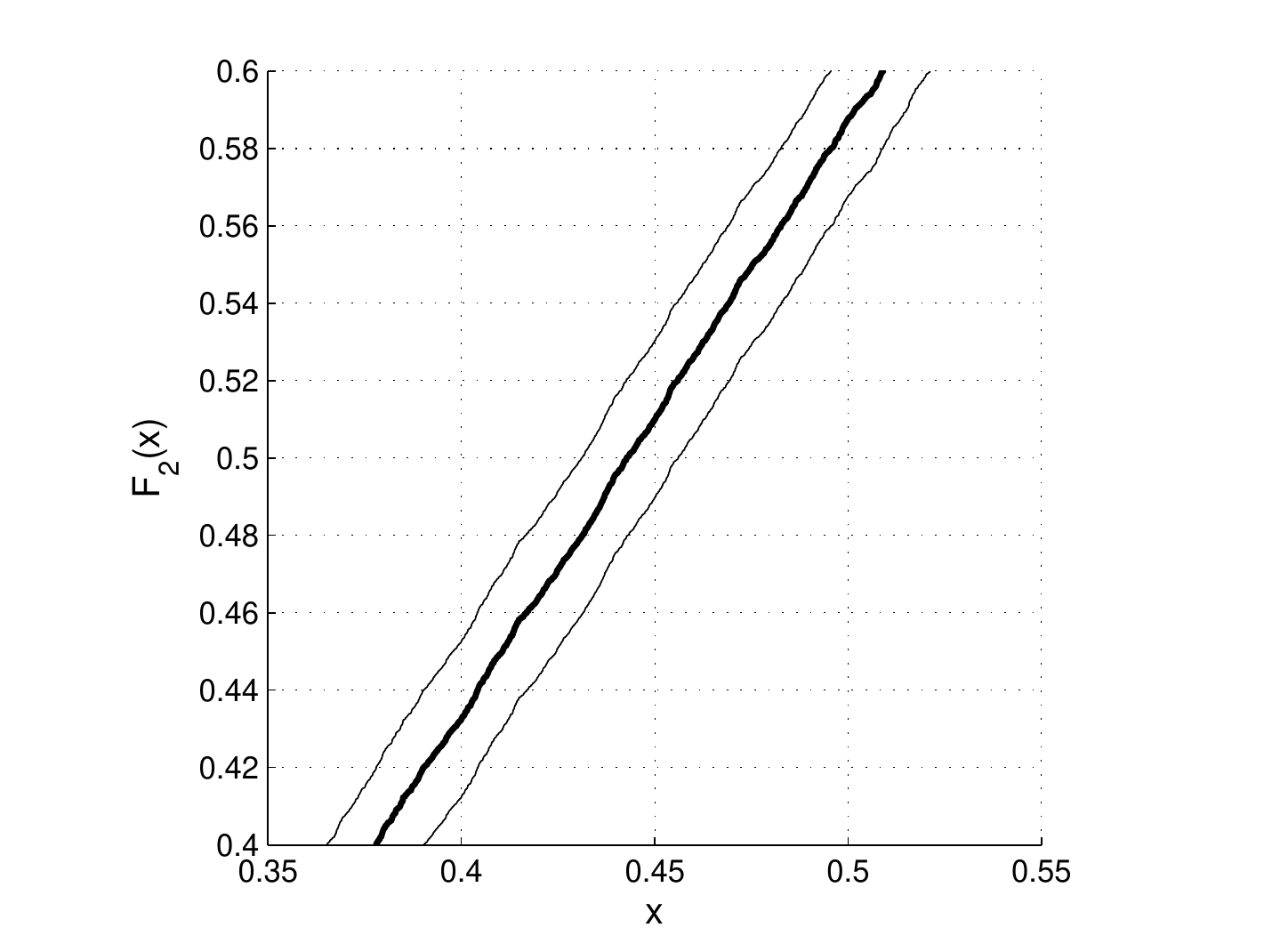}};
  \draw (-0.5,-0.3) -- (0.25,-0.3) -- (0.25,0.5) -- (-0.5,0.5) -- cycle;
  \draw[dashed] (0.25,0.5)  -- (4.25,1.5);
  \draw[dashed] (0.25,-0.3) -- (4.25,-1.4);
 \end{tikzpicture}
 \caption{Empirical CDF of $\eta^{\Tr}$ on $\Q_\C^1$.}
 \label{Fig:F2}
 \end{subfigure}
 \\
 \begin{subfigure}{\textwidth}
 \centering
 \begin{tikzpicture}
  \draw (0,0) node {\includegraphics[scale = 0.45]{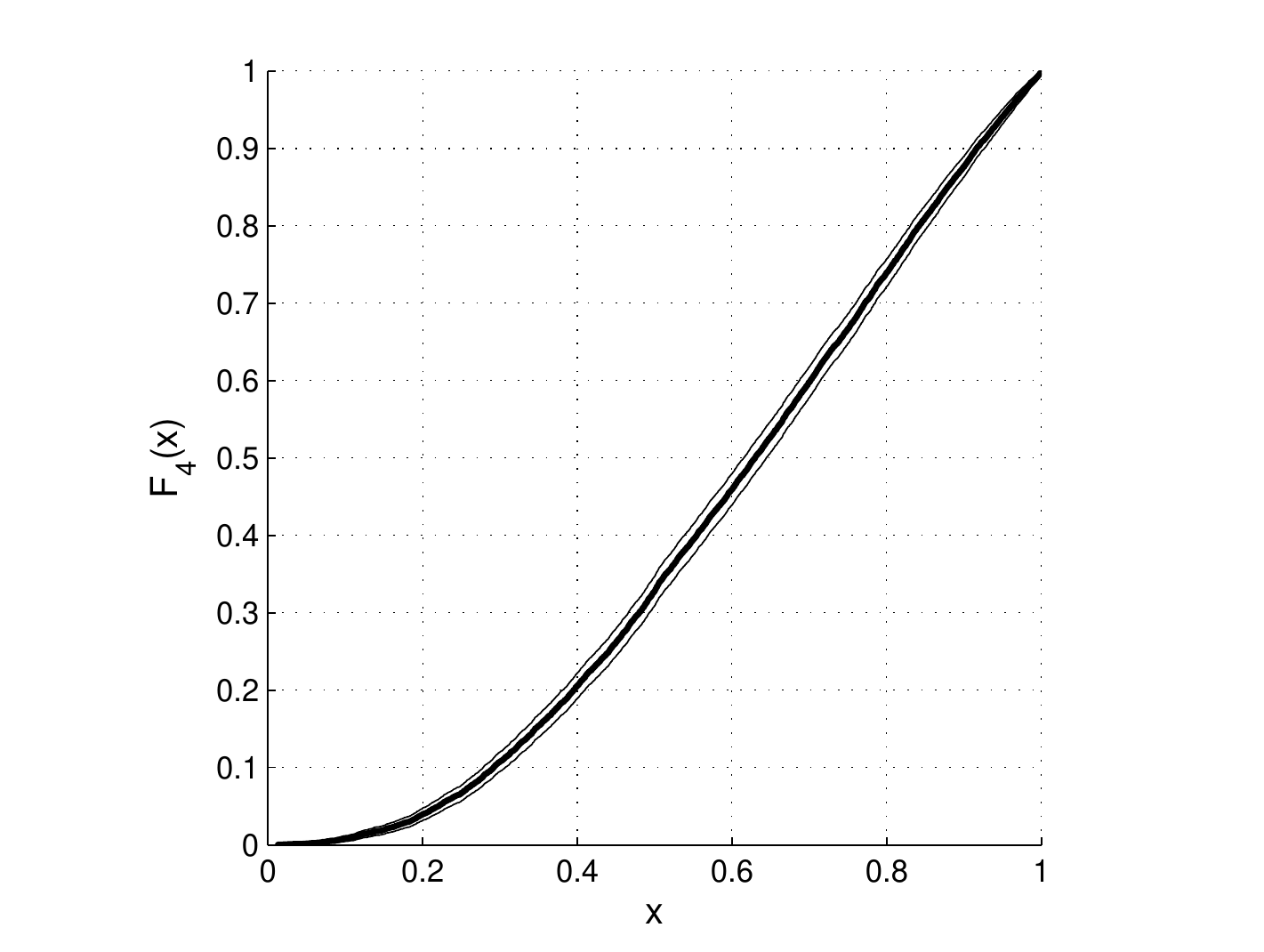}};
  \draw (6,0) node {\includegraphics[scale = 0.35]{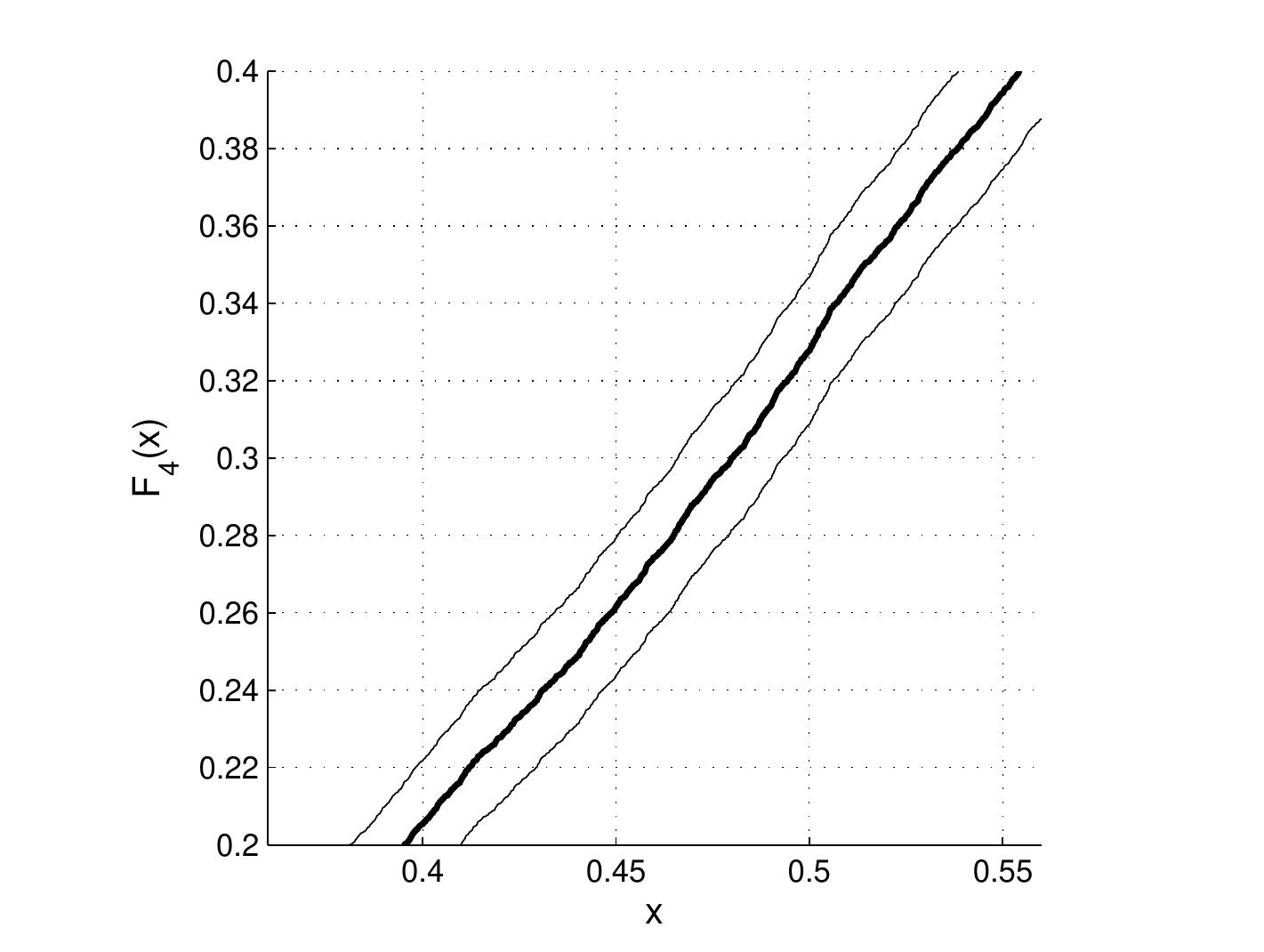}};
  \draw (-0.4,-1.1) -- (0.35,-1.1) -- (0.35,-0.3) -- (-0.4,-0.3) -- cycle;
  \draw[dashed] (0.35,-0.3) -- (4.25,1.5);
  \draw[dashed] (0.35,-1.1) -- (4.25,-1.4);
 \end{tikzpicture}
 \caption{Empirical CDF of $\eta^{\Tr}$ on $\Q_\C$.}
 \label{Fig:F4}
 \end{subfigure}
 \caption{
    Empirical CDF of $\eta^{\Tr}$ and confidence
    band ($n=10^4$, $\alpha = 5\times 10^{-5}$).
 }
 \label{Fig:F1234}
\end{figure}

\subsection{Distribution $\eta^{\Tr}$ over classical channels}

Three natural questions arise about the distribution of
  trace-distance contraction coefficient over a fixed classical channel: 
\begin{enumerate}[i.]
\item What is the supremum of $\eta^{\Tr}$ over a fixed classical channel?
\item  What is the infimum of $\eta^{\Tr}$ over a fixed classical channel?
\item  What is the typical value (the mode) of $\eta^{\Tr}$ 
 over a fixed classical channel?
\end{enumerate}

The set of qubit channels over the classical channel
$ \left(\begin{array}{cc}
        1-a & a\\
        1-f & f
       \end{array}
 \right)$
  with respect to the parametrization \eqref{eq:matQ} is denoted by
  $\Q_\R (a,f)$, $\Q_\C (a,f)$, $\Q_\R^1 (a)$ and $\Q_\C^1 (a)$.
The next Theorem answers the first question.
\begin{theorem}
Let $a,f\in [0,1]$ be arbitrary real numbers.
For all $x\in (|a-f|,\sqrt{(1-a)f}+\sqrt{a(1-f)})$ there exists 
  a qubit channel $Q\in\Q_\R (a,f)\subset \Q_\C (a,f)$ for which 
  $\eta^{\Tr}(Q) = x$.
\end{theorem}
\begin{proof}
Let $a,f\in [0,1]$ and $x\in \left(|a-f|,\sqrt{(1-a)f}+\sqrt{a(1-f)}\right)$
  be arbitrary. 
Consider the following qubit channel 
\begin{equation*}
Q=\left(
\begin{array}{cccc}
 a &   0 & 0 &   d \\
 0 & 1-a & e &   0 \\
 0 &   e & f &   0 \\
 d &   0 & 0 & 1-f
\end{array}
\right),   
\end{equation*}
  where $d,e\in\R$.
In order to guarantee the positivity ot the matrix above, the
  following constrains must be held.
\begin{align*}
 e^2 &\le (1-a)f \\
 d^2 &\le a(1-f) 
\end{align*}
According to \eqref{eq:matT}, 
  $\eta^{\Tr}(Q)=\max \left(|d+e|,|d-e|,|a-f|\right)$, 
  where $|d\pm e|\le |d|+|e|\le \sqrt{(1-a)f}+\sqrt{a(1-f)}$
  which completes the proof.
\end{proof}

\begin{corollary}
For unital channels $f=1-a$ hence the supremum of $\eta^{\Tr} (Q)$
  on the set $\Q_\R^1\subset\Q_\C^1$ is equal to $\sqrt{(1-a)^2}+\sqrt{a^2}=1$
  which means that the theoretical upper bound of $\eta$ can be 
  reached over any classical channel.
\end{corollary}

\begin{conjecture}\label{conj}
We conjecture that $\inf \{\eta (Q):Q\in\Q_\C (a,f)\} = |a-f|$
  which is equal to the trace-distance contraction coefficient of 
  the underlying classical channel.
\end{conjecture}

It seems that there is no chance to give explicit formula for the mode of  
  $\eta^{\Tr}$ over a fixed classical channel.
Instead of this, Monte-Carlo simulations were done for the case of unital 
  channels. 
The interval $[0,1]$ was divided into $100$ equidistant parts.
Infimum, expectation and mode was estimated from a sample of size $n=1000$ in 
  each point.
Infimum of $\eta^{\Tr}$ in $a\in [0,1]$ was estimated by the following formula.
\begin{equation*}
\inf \{\eta^{\Tr}(Q):Q\in\Q^1(a)\}\approx\min (|2a-1|,
\text{smallest element in the sample})
\end{equation*}
Smoothed density histogram was applied to estimate the mode.
Confidence band corresponding to the confidence level $99.995\%$ 
  ($\alpha = 5\times 10^{-5}$) was calculated for the expected value.
\begin{figure}[!ht]
    \centering
    \begin{subfigure}{0.48\textwidth}
        \centering
        \includegraphics[width=\textwidth]{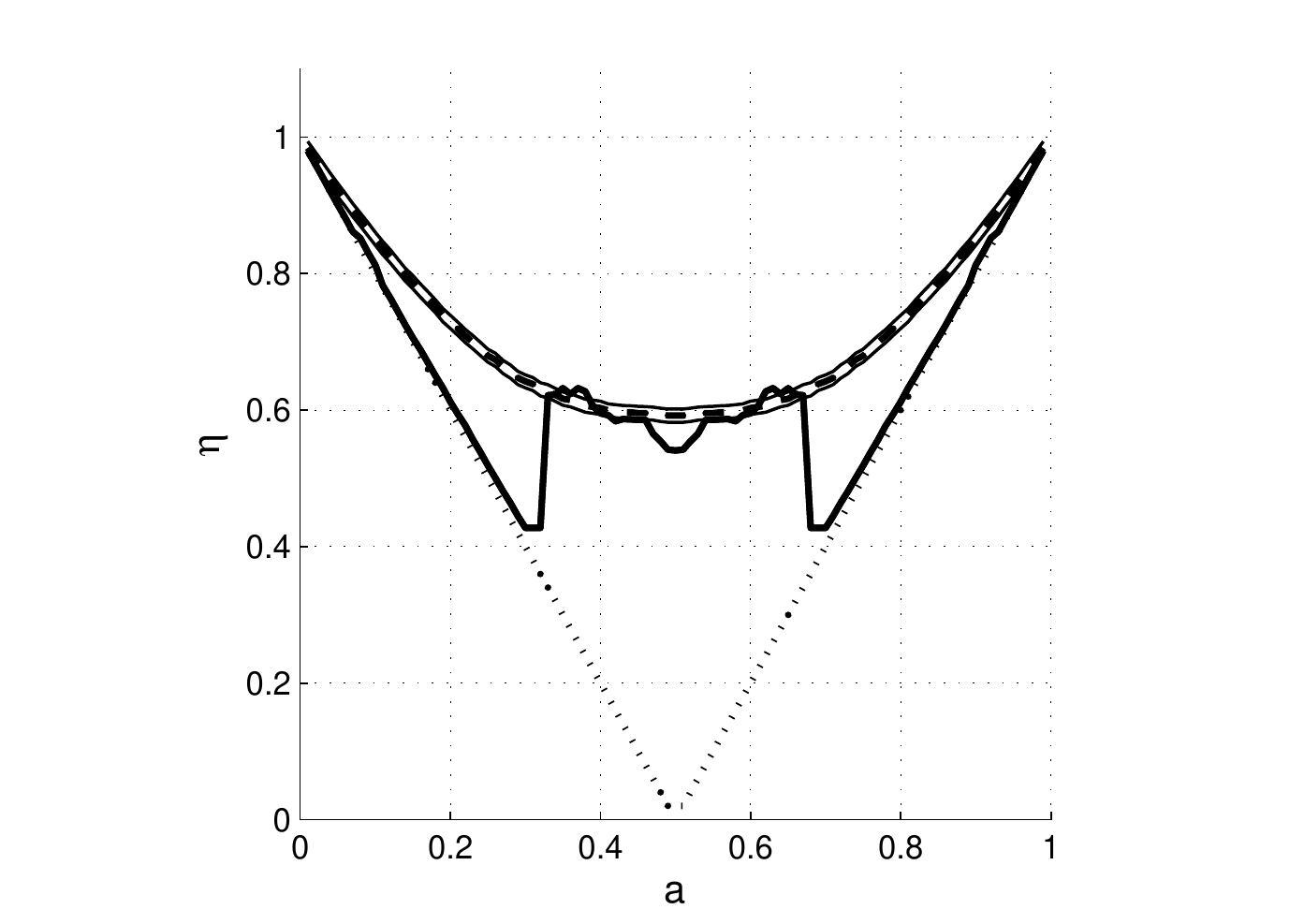}
        \caption{$\Q_\R^1$}
        \label{Fig:etaR}
    \end{subfigure}
    \hfill
    \begin{subfigure}{0.48\textwidth}
        \centering
        \includegraphics[width=\textwidth]{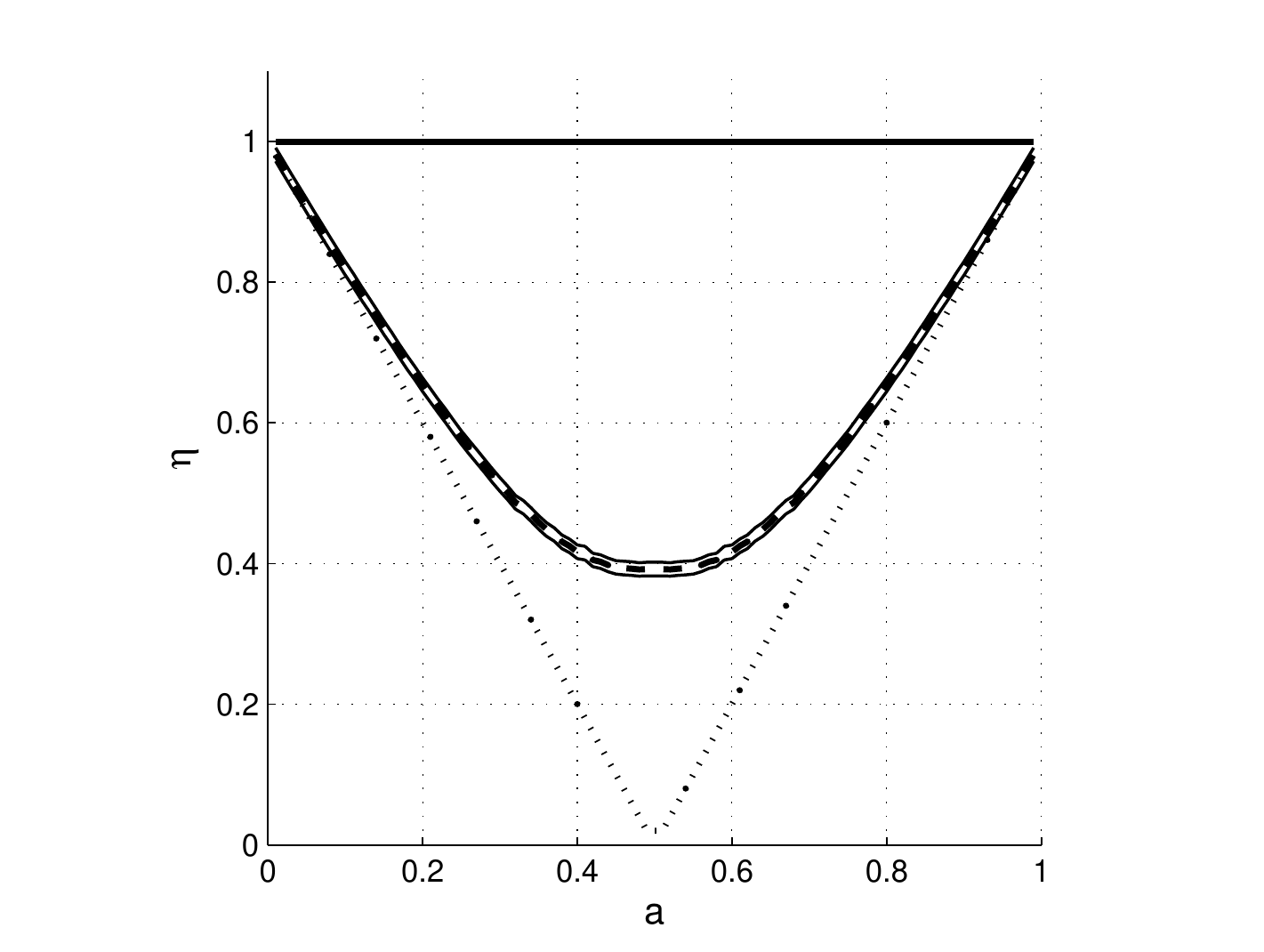}
        \caption{$\Q_\C^1$}
        \label{Fig:etaC}
    \end{subfigure}
    \caption{Minimal value (dotted) of $\eta^{\Tr}$, 
             mode of $\eta^{\Tr}$ (thick),
             expectation of $\eta^{\Tr}$ (dashed) and 
             confidence band (solid) corresponding 
             to the expectation ($n=1000$, $\alpha = 5\times 10^{-5}$).}
    \label{Fig:etaRC}
\end{figure}
The estimated infimum of $\eta^{\Tr}$ is displayed by dotted line in
  Figure \ref{Fig:etaRC}.
We can see that the estimated infimum coincides with the trace-distance 
  contraction coefficient of the underlying classical channel which confirms 
  Conjecture \ref{conj} for unital channels.
The mode shows irregular behaviour in case of real unital channels
  (Figure \ref{Fig:etaR}).
Small deviations of mode from infimum can be observed near $a\approx 0.1$ and
  $a\approx 0.9$ and the distribution of $\eta^{\Tr}$ changes dramatically 
  near $a\approx 0.33$ and $a\approx 0.67$.
We can see in Figure \ref{Fig:etaC} that qubit channels over the complex field 
  are condensed near the extremal $\eta^{\Tr} = 1$ isosurface. 

%\bibliography{Volume} 

\begin{thebibliography}{10}

\bibitem{AndaiVol}
Attila Andai.
\newblock Volume of the quantum mechanical state space.
\newblock {\em J. Phys. A}, 39(44):13641--13657, 2006.

\bibitem{Choi}
Man~Duen Choi.
\newblock Completely positive linear maps on complex matrices.
\newblock {\em Linear Algebra and Appl.}, 10:285--290, 1975.

\bibitem{Cohen}
Joel~E. Cohen, Yoh Iwasa, Gh. R{\u{a}}u{\c{t}}u, Mary~Beth Ruskai, Eugene
  Seneta, and Gh. Zb{\u{a}}ganu.
\newblock Relative entropy under mappings by stochastic matrices.
\newblock {\em Linear Algebra Appl.}, 179:211--235, 1993.

\bibitem{Kemperman}
Joel~E. Cohen, J.~H.~B. Kemperman, and Gh. Zb{\u{a}}ganu.
\newblock {\em Comparison of Stochastic Matrices with Applications in
  Information Theory, Statistics, Economics and Population Sciences}.
\newblock Birkh\"auser, Boston, 1998.

\bibitem{Dobrushin}
R.~L. Dobrushin.
\newblock Central limit theorem for nonstationary markov chains. ii.
\newblock {\em Theory of Probability \& Its Applications}, 1(4):329--383, 1956.

\bibitem{Grandshteyn}
I.~S. Gradshteyn, I.~M. Ryzhik, A.~Jeffrey, and D.~Zwillinger.
\newblock {\em Tables of Integrals, Series, and Products}.
\newblock CA: Academic Press, San Diego, 2000.

\bibitem{Kastoryano}
Michael~J. Kastoryano and Kristan Temme.
\newblock Quantum logarithmic {S}obolev inequalities and rapid mixing.
\newblock {\em J. Math. Phys.}, 54(5):052202, 30, 2013.

\bibitem{Knuth}
Donald~E. Knuth.
\newblock {\em The art of computer programming. {V}ol. 2}.
\newblock Addison-Wesley, Reading, MA, 1998.
\newblock Seminumerical algorithms, Third edition [of MR0286318].

\bibitem{Lawless}
Jerald~F. Lawless.
\newblock {\em Statistical models and methods for lifetime data}.
\newblock Wiley Series in Probability and Statistics. Wiley-Interscience [John
  Wiley \& Sons], Hoboken, NJ, second edition, 2003.

\bibitem{NielsenChuang}
Michael~A. Nielsen and Isaac~L. Chuang.
\newblock {\em Quantum computation and quantum information}.
\newblock Cambridge University Press, Cambridge, 2000.

\bibitem{IgorPak}
Igor Pak.
\newblock Four questions on {B}irkhoff polytope.
\newblock {\em Ann. Comb.}, 4(1):83--90, 2000.

\bibitem{PetzQinf}
D{\'e}nes Petz.
\newblock {\em Quantum information theory and quantum statistics}.
\newblock Theoretical and Mathematical Physics. Springer-Verlag, Berlin, 2008.

\bibitem{RuskaiSzarekWerner}
Mary~Beth Ruskai, Stanislaw Szarek, and Elisabeth Werner.
\newblock An analysis of completely positive trace-preserving maps on
  {$\mathcal{M}_2$}.
\newblock {\em Linear Algebra Appl.}, 347:159--187, 2002.

\bibitem{Temme}
K.~Temme, M.~J. Kastoryano, M.~B. Ruskai, M.~M. Wolf, and F.~Verstraete.
\newblock The {$\chi^2$}-divergence and mixing times of quantum {M}arkov
  processes.
\newblock {\em J. Math. Phys.}, 51(12):122201, 19, 2010.

\end{thebibliography}
%\bibliographystyle{plain}
% Copy from: volume04.bbl

\end{document}